\documentclass[]{template}

\graphicspath{{figs/}}

\usepackage{amsmath,mathtools, epigraph, framed, mdframed, xcolor}
\usepackage{algorithm}
\usepackage{algpseudocode}
\usepackage{amsthm}
\usepackage{graphicx}
\usepackage{eucal}
\usepackage[utf8]{inputenc}

\newcommand*{\Scale}[2][4]{\scalebox{#1}{\ensuremath{#2}}}%

\theoremstyle{definition}
\newtheorem{definition}{Definition}[section]
\newtheorem{theorem}{Theorem}[section]

\begin{document}

% Larger bottom margin for the first page
\newgeometry{bottom=1.5in}

\begin{center}

  %\title{\textsc{\huge Spes contra spem}\\Perebor evitabilis est}
  \title{\textsc{The Long, the Short and the Random}}
  %\date{\today}
  \maketitle 

  % Start page numbering on second page. Must appear *after* \maketitle
  \thispagestyle{empty}
  
  \vspace*{.2in}

  % Authors and Affiliations
  \begin{tabular}{cc}
    Giorgio Camerani\upstairs{\affilone}%, Second Author\upstairs{\affiltwo}, Third Author\upstairs{\affilthree}
   \\[0.25ex]
   {\small Rome, Italy - 29 September 2020}\\
   %{\small \upstairs{\affiltwo} Affiliation Two} \\
   %{\small \upstairs{\affilthree} Affiliation Three} \\
  \end{tabular}
  \date{\today}
  % Replace with corresponding author email address
  \emails{
    \upstairs{\affilone}giorgio.camerani@gmail.com 
    }
  \vspace*{0.4in}

\textcolor[RGB]{50,50,50}{\rule{290pt}{0.1pt}}
\begin{abstract}
We furnish solid evidence, both theoretical and empirical, towards the existence of a deterministic algorithm for random sparse $\#\Omega(\log n)$-SAT instances, which computes the exact counting of satisfying assignments in sub-exponential time. The algorithm uses a nice combinatorial property that every CNF formula has, which relates its number of unsatisfying assignments to the space of its monotone sub-formulae.
\end{abstract}
\end{center}

\vspace*{0.15in}
\hspace{18.5pt}
  \Small	

%\textcolor[RGB]{220,220,220}{\rule{\linewidth}{0.2pt}}

\vspace*{0.5in}
\section{Introduction}
\label{sec:intro}
\noindent Let $\Phi = \Phi( n, \delta, \lambda, \Lambda_{\downarrow}, \Lambda_{\uparrow} )$ be a $k$-CNF formula on $n$ boolean variables with $\delta n$ clauses, each one having length $k$ such that $\Lambda_{\downarrow} \leq k \leq \Lambda_{\uparrow}$ and being chosen uniformly at random among the $\sum_{k = \Lambda_{\downarrow}}^{\Lambda_{\uparrow}} 2^k {n \choose k}$ possible candidates, where $\delta \geq 1$ and $\lambda \geq 1$ are constants, where $\Lambda_{\downarrow}(n)$ is any function such that $\Lambda_{\downarrow}(n) \geq \lambda \log n$, and where $\Lambda_{\uparrow}(n)$ is any function such that $\Lambda_{\downarrow}(n) \leq \Lambda_{\uparrow}(n) \leq n$. In this paper, we present a deterministic sub-exponential time algorithm for computing the exact number of satisfying assignments of any such $\Phi$.
\\\\
Our algorithm, which is very simple in its essence, takes profit from a nice combinatorial property which holds for any CNF formula, no matter if sparse or dense, no matter how long its clauses are allowed to be nor whether their length is the same for all rather than different, and no matter if random or not. Just as general a property as possible. It is precisely the random nature of $\Phi$ however, together with the at least logarithmic length of all its clauses, which combines with such property in such a way that renders us able to obtain sub-exponential running time. If $\Phi$ was not random, that property would not have brought any tangible benefit, leaving us with an exponential time algorithm only. While if $\Phi$ was random, but with constant clause length, again we would only have had an exponential time algorithm, albeit with a decreasing exponent as $k$ increased.
\\\\
Yes, this. As surprising as it may seem, and as contrary to conventional wisdom as it can be, believe it or not, that combinatorial property, which relates the number of unsatisfying assignments a CNF formula has to the landscape of its monotone sub-formulae, does, in the case of random formulae, lead us to the following inevitable conclusion: the longer the clause length, the shorter the running time; the longer the clause length, the shorter the counting certificate.
\newpage
\subsection{Contents}The rest of this paper is organized as follows:
\\\\
\begin{tabular}{p{2cm}p{14cm}}
\textbf{Section 2} & Illustrates the combinatorial property our algorithm relies upon. \\
 & \\
\textbf{Section 3} & Describes the algorithm and the intuition behind it. \\
 & \\
\textbf{Section 4} & Proves a sub-exponential upper bound on its running time. \\
 & \\
\textbf{Section 5} & Presents empirical evidence in further support of the theoretical proof. \\
 & \\
\textbf{Section 6} & Ends the document by identifying what seem to be promising avenues of future work. \\
\end{tabular}
\\
\section{The property: from unsatisfying assignments to monotone sub-formulae}\label{sec:property}
\noindent We begin by introducing some classic notation and definitions, then we will proceed straight to the illustration of our combinatorial property, which allows us to express the number of unsatisfying assignments as a function of the landscape of monotone sub-formulae.

\subsection{Notations and definitions} Let $\Phi = \{c_1, \cdots, c_m\}$ be a \textit{generic} CNF formula\footnote{In this section \ref{sec:property}, and only here, $\Phi$ will be generic. In all the rest of the paper, $\Phi$ will be what we have defined in section \ref{sec:intro}.} with $n$ variables and $m$ clauses. Let $V = \{v_1, \cdots, v_n\}$ be the set of variables\footnote{Without loss of generality, we will assume that $V$ is induced by the clauses in $\Phi$. Formally, $\forall v \in V \ \exists c \in \Phi : v \in c \lor \lnot v \in c$. In other words, there cannot be variables which are not mentioned anywhere in $\Phi$: its variables are all and only those that appear in its clauses.} of $\Phi$. Each clause $c_i = \{\ell_{i,1}, \cdots, \ell_{i,|c_i|}\}$ is a set of literals, where each literal is either a variable $v \in V$ or its negation. The cardinality $|c_i|$ of a clause $c_i$ is also known as its \textit{length}. Let $\mathcal{A} = \{v_1, \lnot v_1\} \times \cdots \times \{v_n, \lnot v_n\}$ denote the set of all the $2^n$ possible boolean assignments to the $n$ variables in $V$. Let $\mathcal{S} = \{ b \in \mathcal{A} : \forall c \in \Phi \ c \cap b \neq \varnothing \}$ be the set of satisfying assignments of $\Phi$. Let $\mathcal{U} = \mathcal{A} \setminus \mathcal{S} = \{ b \in \mathcal{A} : \exists c \in \Phi \ c \cap b = \varnothing \}$ be the set of unsatisfying assignments of $\Phi$.

\theoremstyle{definition}
\begin{definition}[\textit{Sub-formula} of $\Phi$]
A \textit{sub-formula} $\Psi$ of $\Phi$ is any formula $\Psi \subseteq \Phi$.
\end{definition}

\theoremstyle{definition}
\begin{definition}[\textit{Monotone} formula]
A formula is \textit{monotone} if and only if each of its variables always appears with the same sign: either always positive or always negated. 
\end{definition}

\noindent See how for a formula to be monotone it is not required that all the variables carry the same sign. Different variables can have different signs. The only restriction is that every same variable always carries the same sign. This paragraph culminates by printing the following truism:
\begin{equation}\label{eq:truism}
|\mathcal{S}| = 2^n - |\mathcal{U}|
\end{equation}

\subsection{The property} Let $O_\nu$ be the number of monotone sub-formulae of $\Phi$ having $\nu$ variables and an odd number of clauses. Let $E_\nu$ be the number of monotone sub-formulae of $\Phi$ having $\nu$ variables and an even number of clauses. Our property consists in the following: 
\begin{theorem}
\begin{equation}\label{eq:property}
|\mathcal{S}| = 2^n - \sum_{\nu = 1}^{n} ( O_\nu - E_\nu ) \cdot 2^{n-\nu}    
\end{equation}
\end{theorem}
\begin{proof}
For each clause $c = \{\ell_{1}, \cdots, \ell_{|c|}\} \in \Phi$, we define the set $\mathcal{U}_c$ as follows:
\[
\mathcal{U}_c = \{ b \in \mathcal{A} : c \cap b = \varnothing \}
\]
$\mathcal{U}_c$ contains all those boolean assignments that cannot be satisfying for $\Phi$, due to the presence of  $c$ itself: observe how every clause $c$, just because it exists, has the effect of putting outside $\mathcal{S}$, and thus in $\mathcal{U}$, all those $b \in \mathcal{U}_c$. And which are those $b$ exactly? Well, it is easy to see that they are all and only those $2^{n - |c|}$ boolean assignments having all the literals $\ell_{1}, \cdots, \ell_{|c|}$ set to false. The set $\mathcal{U}$ can then be expressed as:
\[
\mathcal{U} = \bigcup_{i = 1}^{m} \mathcal{U}_{c_i}
\]

\theoremstyle{definition}
\begin{definition}[\textit{Conflicting} clauses]
Any $2$ clauses $c_1, c_2 \in \Phi$ are said to be \textit{conflicting} if and only if $\exists v \in V : v \in c_1 \land \lnot v \in c_2$. 
\end{definition}

\noindent It is easy to see that for any $2$ clauses $c_1, c_2 \in \Phi$, it is the case that $\mathcal{U}_{c_1} \cap \mathcal{U}_{c_2} \neq \varnothing$ if and only if $c_1$ and $c_2$ are not conflicting. For if they are, any $b \in \mathcal{U}_{c_1}$ will satisfy $c_2$ thus will not be in $\mathcal{U}_{c_2}$, and vice-versa\footnote{That is to say, it is impossible to falsify the one without satisfying the other.}, hence $\mathcal{U}_{c_1}$ and $\mathcal{U}_{c_2}$ would have empty intersection. Whereas if they are not, $\mathcal{U}_{c_1} \cap \mathcal{U}_{c_2}$ would contain all and only those $2^{n-|c_1 \cup c_2|}$ boolean assignments falsifying both $c_1$ and $c_2$, that is to say having all literals of $c_1$ set to false and all literals of $c_2$ set to false. Now, as we have to determine the cardinality of $\mathcal{U}$, which is a set defined as the union of possibly non-disjoint sets, we invoke the inclusion-exclusion principle:

\begin{equation}
\label{eq:inclusionexclusion}
|\mathcal{U}| = \biggl|\bigcup_{i = 1}^{m} \mathcal{U}_{c_i}\biggr| = \sum_{i = 1}^{m} (-1)^{i+1} \cdot \left( \sum_{1 \leq j_1 < \cdots < j_i \leq m} | \mathcal{U}_{c_{j_1}} \cap \cdots \cap \mathcal{U}_{c_{j_i}} | \right)
\end{equation}
\noindent Observe how the inner summation scans the \textit{monotone} sub-formulae of $\Phi$ having $i$ clauses. Why monotone? Because if the sub-formula $\Psi$ of $\Phi$ obtained by picking clauses $c_{j_1}, \cdots, c_{j_i}$ is not monotone, then there must be $c_1, c_2 \in \Psi$ such that $c_1$ and $c_2$ are conflicting. But if they are, then $\mathcal{U}_{c_1} \cap \mathcal{U}_{c_2} = \varnothing$, and the entire term $| \mathcal{U}_{c_{j_1}} \cap \cdots \cap \mathcal{U}_{c_{j_i}} |$ evaluates to $0$. This means that non-monotone sub-formulae play no role here: they bring no contribution to the inner summation. Only monotone sub-formulae do matter. On the other hand, the outer summation scans the possible numbers of clauses a sub-formula may have, from just $1$ to $m$: note how those sub-formulae having an odd number of clauses give an additive contribution, whereas those having an even number of clauses give a subtractive contribution.
\\\\
Let us focus on the generic term of the inner summation: the cardinality of the set $\mathcal{U}_{c_{j_1}} \cap \cdots \cap \mathcal{U}_{c_{j_i}}$. By applying the definition of $\mathcal{U}_c$, we can write:
\begin{equation*}
\mathcal{U}_{c_{j_1}} \cap \cdots \cap \mathcal{U}_{c_{j_i}} = \{ b \in \mathcal{A} : c_{j_1} \cap b = \varnothing \land \cdots \land c_{j_i} \cap b = \varnothing \}
\end{equation*}
It is now easy to see that the set $\mathcal{U}_{c_{j_1}} \cap \cdots \cap \mathcal{U}_{c_{j_i}}$ contains all and only those boolean assignments that falsify all the clauses of the monotone sub-formula $\Psi$ obtained by picking clauses $c_{j_1} \cdots c_{j_i}$ from $\Phi$. Therefore its cardinality is expressed by the following equation:
\begin{equation*}
|\mathcal{U}_{c_{j_1}} \cap \cdots \cap \mathcal{U}_{c_{j_i}}| = 2^{n-\nu}
\end{equation*}
where $\nu$ is the number of variables of $\Psi$. Hence the contribution given by the generic term of the inner summation does not depend on which nor on how many clauses are picked, but only on the number of variables of the monotone sub-formula $\Psi$ induced by them. How many clauses are picked is only relevant in the outer summation, to determine if such contribution has to be added or subtracted.
\\\\
Sub-formulae with the same number of variables give exactly the same contribution, regardless of their number of clauses which only affects the sign of such contribution. By grouping sub-formulae according to their number of variables rather than according to their number of clauses, we are thus able to rearrange equation \ref{eq:inclusionexclusion} as follows:

\begin{equation}
\label{eq:rearrangement}
|\mathcal{U}| = \sum_{\nu = 1}^{n} O_\nu \cdot 2^{n-\nu} - \sum_{\nu = 1}^{n} E_\nu \cdot 2^{n-\nu} = \sum_{\nu = 1}^{n} ( O_\nu - E_\nu ) \cdot 2^{n-\nu}
\end{equation}
The proof ends by observing that substituting equation \ref{eq:rearrangement} into equation \ref{eq:truism} leads to equation \ref{eq:property}.
\end{proof}

\section{The algorithm and the intuition that led to it}\label{sec:intuition}
\noindent From this moment onwards, we will take the assumption, without loss of generality, that all the clauses of our input formula $\Phi$ have length exactly\footnote{Each time we write $\lambda \log n$ we actually mean $\left \lceil{\lambda \log n}\right \rceil$. The same holds for $\delta n$ and $\left \lceil{\delta n}\right \rceil$ of course. Such alleviation of notation is meant to not encumber the eyes of the reader, nor the hands of the writer.} $\lambda \log n$, that is to say $\Lambda_{\downarrow} = \Lambda_{\uparrow} = \lambda \log n$. Such assumption is not going to hamper the validity of our reasoning, though, because it will only worsen the running time of the algorithm: as we will see in section \ref{sec:proof}, once sub-exponential running time is proven for a certain clause length, it can only get better when the length is augmented. Let us recall the \textit{leitmotiv} of the whole proof:
\\
\begin{mdframed}[backgroundcolor=gray!7] 
\vspace*{\fill} 
\begin{quote} 
\centering 
\textsc{\Small Longer means Shorter} 
\end{quote}
\vspace*{\fill}
\end{mdframed}
\vspace{10pt}
\noindent 
which in turn implies that what we will prove under the $\Lambda_{\downarrow} = \Lambda_{\uparrow} = \lambda \log n$ assumption will be \textit{a fortiori ratione} valid in the more general case where $\Lambda_{\downarrow} \geq \lambda \log n$.
\\\\
Without further ado, let's start. The intuition behind the algorithm is epitomized by the following question:
\\
\begin{mdframed}[backgroundcolor=gray!7] 
\vspace*{\fill} 
\begin{quote} 
\centering 
\textit{What if the set of \textbf{all} monotone sub-formulae of $\Phi$ is sub-exponentially sized?} 
\end{quote}
\vspace*{\fill}
\end{mdframed}
\vspace{10pt}
\noindent 
If this is indeed the case, we might just enumerate such whole set, kind of plugging it into equation \ref{eq:rearrangement}. Doing so, we would count the unsatisfying assignments of $\Phi$ in sub-exponential time. For $\nu = \lambda \log n, \cdots, n$ let $\langle \nu, O_\nu, E_\nu \rangle$ be a sequence of triples, each one indicating that there are $O_\nu$ (respectively $E_\nu$) monotone sub-formulae of $\Phi$ having $\nu$ variables and an odd (respectively even) number of clauses. We would just scan the entire set of monotone sub-formulae: each time we find one having $\nu$ variables, we update the corresponding triple by increasing by $1$ either $O_\nu$ or $E_\nu$, depending on the parity of its number of clauses. Once we have finished to scrutinize the whole space of monotone sub-formulae, we would have that sequence of $n - \lambda \log n + 1$ triples ready to be plugged into the summation of equation \ref{eq:property}.
\\\\
But what does it mean for that set to be sub-exponentially sized? Intuitively, it means that there exists a (sufficiently small) maximum number of clauses, let us call it $i_{STOP}$, that any monotone sub-formula of $\Phi$ can possibly have. Above such number, nothing else would exist to be perlustrated: there would be no monotone sub-formula having more than $i_{STOP}$ clauses. Asserting that such set is sub-exponentially sized is then equivalent to assert that $i_{STOP} \in o(n)$. So, our algorithm would simply enumerate all monotone sub-formulae having up to $i_{STOP}$ clauses, knowing that no further others can exist above such limit:
\\
\begin{mdframed}[backgroundcolor=gray!7] 
\vspace*{\fill} 
\begin{quote} 
\centering 
\textit{Just brute force, yes. But applied to a sub-exponentially sized space.} 
\end{quote}
\vspace*{\fill}
\end{mdframed}
\vspace{10pt}
\newpage \noindent
Here is the pseudocode of our algorithm, lines from $3$ to $9$ correspond to the brute force enumeration of all monotone sub-formulae of $\Phi$ having $i \leq i_{STOP}$ clauses (that is, all of them), whereas lines from $10$ to $14$ correspond to equation \ref{eq:property}:
\renewcommand{\thealgorithm}{}
\begin{algorithm}
\caption{Computes the number of satisfying assignments of $\Phi$}
\label{alg:count}
\begin{algorithmic}[1]
\Procedure{Count}{$\Phi$}
    \State Initialize $\langle \nu, O_{\nu}, E_{\nu} \rangle \gets \langle \nu, 0, 0 \rangle$, $\forall \nu \in [\ \lambda \log n,\ n\ ]$
    %\For {$i \in [\ 1,\ i_{STOP}\ ]$}
        \For {each monotone sub-formula of $\Phi$ having $i \leq i_{STOP}$ clauses and $\nu$ variables}
            \If {$i$ is odd}
                \State $\langle \nu, O_{\nu}, E_{\nu} \rangle \gets \langle \nu, O_{\nu} + 1, E_{\nu} \rangle$
            \Else 
                \State $\langle \nu, O_{\nu}, E_{\nu} \rangle \gets \langle \nu, O_{\nu}, E_{\nu} + 1 \rangle$
            \EndIf
        \EndFor
	%\EndFor
	\State $count \gets 0$
	\For{$\nu \in [\ \lambda \log n,\ n\ ]$}
	    \State $count \gets count + ( O_{\nu} - E_{\nu} ) \cdot 2^{n-\nu}$
	\EndFor
	\State Return $2^n-count$
\EndProcedure
\end{algorithmic}
\end{algorithm}

\noindent
Let us formulate one quick inspirational thought:\\
\begin{mdframed}[backgroundcolor=gray!7] 
\vspace*{\fill} 
\begin{quote} 
\centering 
\textit{We are counting without searching} 
\end{quote}
\vspace*{\fill}
\end{mdframed}
\vspace{10pt}
We are not even trying to search for satisfying assignments: no DPLL backtracking search, no clause learning, no symmetry breaking, no walkings, no simulated annealing, no random restarts, nothing of all of this. We are not even looking at the space of $2^n$ possible boolean assignments, let alone walking inside it. We simply completely ignore the exponentially sized space of satisfying assignments, and infer how many are there by judging from the inspection of the landscape of sub-exponentially many monotone sub-formulae. We take a completely different route, in a new paradigm shift which makes us able to shortcut it all, avoiding the burden of wandering within that hostile, exponentially wide territory. Just one further reflection, about such $2$ spaces: not only their size is fundamentally different, also their nature is. While it is hard to even find a single $1$ satisfying assignment, let alone count them by jumping from one cluster to the next, the picture of monotone sub-formulae is a completely different story: it is effortless to find them, they are just there awaiting to be enumerated.
\\\\\noindent
Back to intuition again. Let us now ask the following question: why should we, when $\Phi$ is random and its clauses have at least logarithmic length, believe that the space of all its monotone sub-formulae is going to be sub-exponentially sized?

\theoremstyle{definition}
\begin{definition}[\textit{Maximal} monotone sub-formula]
A monotone sub-formula $\Psi$ of $\Phi$ is \textit{maximal} if and only if adding any further clause $c \in \Phi$ to it would render it no longer monotone. 
\end{definition}

\noindent Imagine to randomly pick a maximal monotone sub-formula $\Psi$ of $\Phi$. How? What does it mean to randomly pick it? It means to start with an empty sub-formula $\Psi$, and to randomly pick clauses from $\Phi$, one after the other: each time a clause is picked out, it gets removed from $\Phi$ and, if possible, added to $\Psi$. What does "if possible" mean? It means that the picked clause gets added to $\Psi$ if and only if it does not destroy its monotonicity: if $\Psi$ would remain monotone, then such clause gets added to it, otherwise it is overthrown. If our intuition is not wrong, then the fact that each clause has at least logarithmic length should imply that we will quickly saturate a large fraction of the $n$ available variables. What do we mean by "quickly"? We mean much quicker than what would happen if the clause length was constant. And what do we mean by "large fraction"? We mean at least $\frac{n}{2}$. And what do we mean by "saturate"? Think about this: each time a clause $c$ is added to $\Psi$, the \textit{new} variables that $c$ brings to $\Psi$ get immediately frozen, crystalised: their sign is once and forever established, due to the monotonicity of $\Psi$; therefore each future clause which wants to be enrolled in $\Psi$ shall agree with all the literals of $c$, as well as with all the literals of all the clauses enrolled before $c$. Intuitively, this should imply that, as more and more clauses are added to $\Psi$, the probability that a randomly chosen candidate future clause is compatible with all the so far crystalised variables, which number is growing fast, gets smaller and smaller: this in turn means that the number of clauses of $\Phi$ that we are forced to waste before we find the next one compatible with the crystalised variables gets higher and higher. Until we quickly exhaust all the $\delta n$ clauses originally present in $\Phi$. At such point, this process ends, $\Psi$ will be a randomly picked maximal monotone sub-formula of $\Phi$, and the size $i$ of $\Psi$ (i.e. its number of clauses) will be at most $i_{STOP}$. 
% 1. When figure position is crucial, use the [H] tag 
%    to enforce absolute positioning
% 2. Image files should be referenced without file extension
\begin{figure}[H]
    \centering
    \includegraphics[scale=0.15]{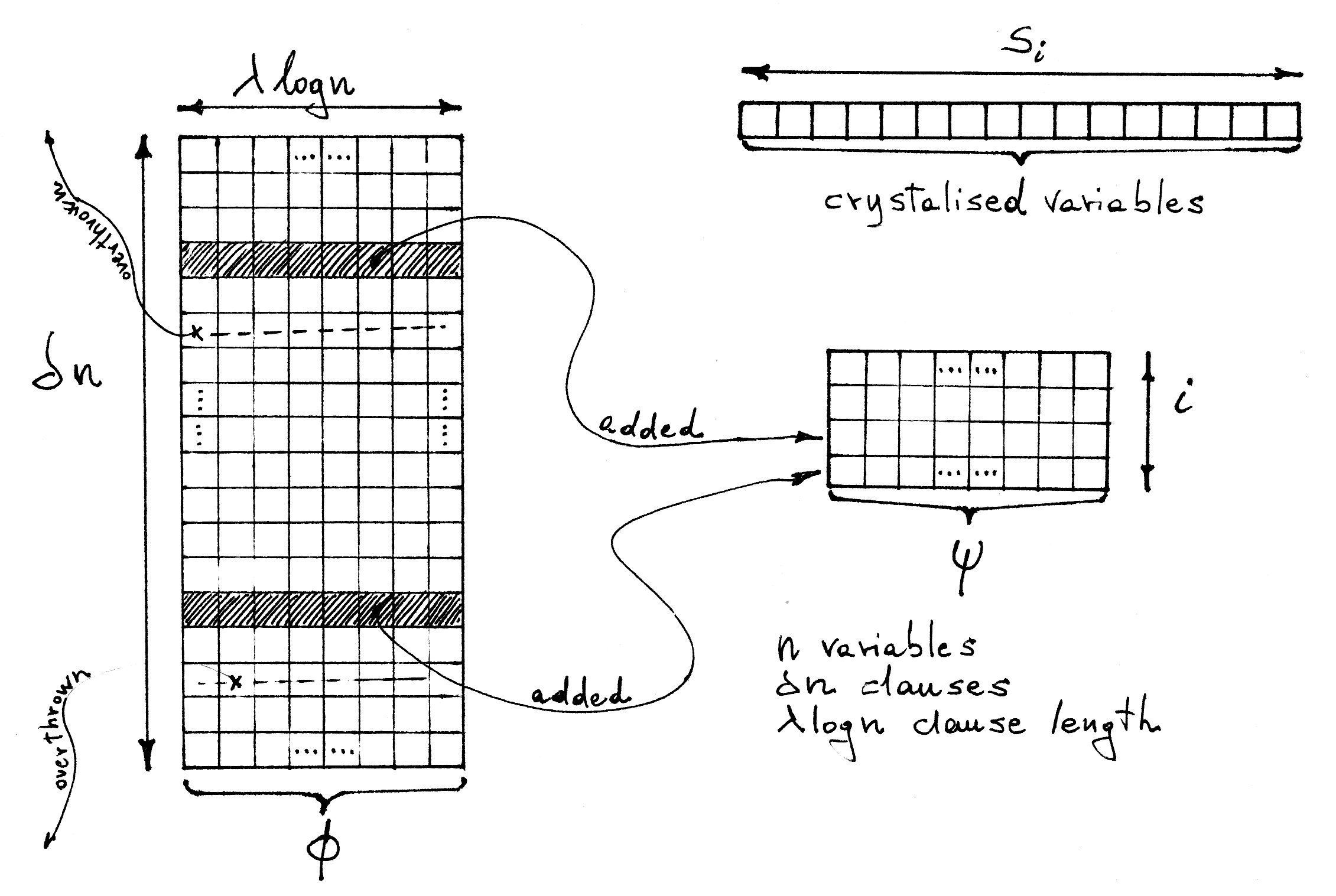} \\
    \caption{Process of randomly picking a maximal monotone sub-formula $\Psi$ of $\Phi$}
    \label{fig:processo_psi}
\end{figure}
\noindent Let us summarize it: the faster the set of crystalised variables grows in size $\rightarrow$ the quicker the probability of picking a compatible clause approaches zero $\rightarrow$ the sooner we will waste all the original clauses at our disposal. If the wasting of clauses happens fast enough, we will run out of them soon enough, when the size of $\Psi$ is still $o(n)$. Epitomizing such intuition:\\   
\begin{mdframed}[backgroundcolor=gray!7] 
\vspace*{\fill} 
\begin{quote} 
\centering 
\textit{We shall dilapidate all the clauses of $\Phi$ quick enough for $i_{STOP}$ to be still sub-linear} 
\end{quote}
\vspace*{\fill}
\end{mdframed}
\vspace{10pt}
See how all of the above narration could not be possible if the clause length was constant: for the size of the set of crystalised variables would have grown too slow for $i_{STOP}$ to remain tamed down at $o(n)$. Also, \textit{ça va sans dire}, no way the trick would have worked if $\Phi$ was not random: reason is, as it will become more clear within the proof in section \ref{sec:proof}, that we need clauses having roughly half of their literals positive and half negative, a set of crystalised variables exhibiting kind of the same rough balancing in their signs, and variables in $\Phi$ occurring about the same number of times, around half of the time positive and half negated. These are all properties of random formulae. As an obvious counterexample, imagine if $\Phi$ was itself monotone: $i_{STOP}$ would have been equal to $\delta n$.
\newpage
\subsection{Formalizing the intuition} We are now going to formalize the process we have just described, which is illustrated in Figure \ref{fig:processo_psi}. To do so, let us first define the following functions:

\begin{center}
\renewcommand{\arraystretch}{2}
\begin{tabular}{ m{0.5cm}| m{12cm} }
 $s_i$ & The expected value of the number of variables of $\Psi$ (i.e. what we have called the crystalised variables) when $\Psi$ already contains $i$ clauses. \\ 
 \hline
 $p_i$ & The probability that a clause $c$ randomly picked from $\Phi$ is \textit{compatible} with $\Psi$ when $\Psi$ already contains $i$ clauses. In other words, it is the probability to insert the $i+1$ clause into $\Psi$. Recall that $c$ being compatible with $\Psi$ means that adding $c$ to $\Psi$ does not destroy $\Psi$'s monotonicity. Here we are basically asking: knowing that $\Psi$ contains $s_i$ variables, which is the probability that the $\lambda \log n$ literals of the candidate clause $c$ do \textit{all} agree with the signs of such $s_i$ variables? \\  
\hline
 $w_i$ & The expected value of the number of clauses of $\Phi$ that were necessary for us to dilapidate so far in order to have $i$ of them into $\Psi$.    
\end{tabular}
\end{center}
Let us first concentrate on $s_i$. We start by observing the following $2$ obvious facts:
\begin{itemize}
   \item  When $\Psi$ contains $0$ clauses, it also contains $0$ variables, so $s_0 = 0$.
   \item  When $\Psi$ contains $1$ clause, it certainly contains exactly $\lambda \log n$ variables, so $s_1 = \lambda \log n$.
 \end{itemize}
What about a generic $i \geq 2$? Imagine to find yourself one moment before adding the $i$-th clause to $\Psi$: in such moment $\Psi$ would have $i-1$ clauses and $s_{i-1}$ variables. Here we have to imagine that our randomly picked candidate $i$-th clause $c$ has already passed the check: it is compatible with the crystalised variables. We are just one instant before adding it to $\Psi$. Now we focus on this question: how many literals $c$ has in common with $\Psi$? Clearly, the number $j$ of such common literals might span from a minimum of $0$ to a maximum of $\lambda \log n$:
\begin{figure}[H]
    \centering
    \includegraphics[scale=0.15]{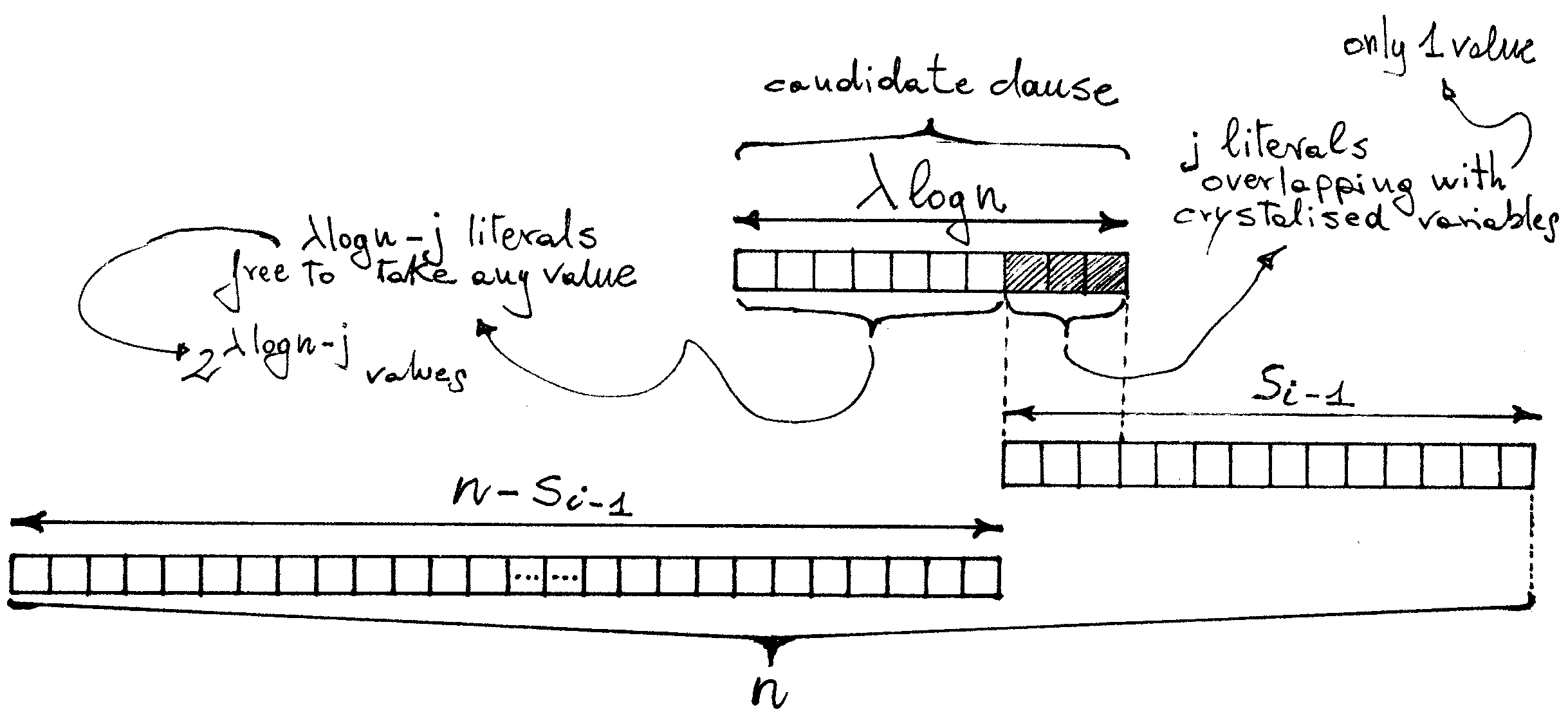} \\
    \caption{$c$ might have $j \in [\ 0,\ \lambda \log n\ ]$ literals overlapping with $\Psi$}
    \label{fig:s_i}
\end{figure}
\noindent Each value $j$ has a certain probability $\pi_j$ of happening. The expected value of $s_i$, i.e. the expected number of variables $\Psi$ will have after having inserted $c$ in it, is then given by the following expression:
\begin{equation*}
s_i = \sum_{j = 0}^{\lambda \log n} ( s_{i-1} + \lambda \log n - j ) \cdot \pi_{j}
\end{equation*}
We now devote our attention to computing the probability $\pi_j$. We will compute it as the fraction of favourable outcomes over the total number of possible outcomes. Clearly, there are $s_{i-1} \choose j$ ways of choosing such $j$ literals among the $s_{i-1}$ many crystalised ones, and there are $n-s_{i-1} \choose \lambda \log n - j$ ways of choosing the remaining others among the not yet crystalised variables. Those $j$ literals cannot oscillate among their $2^j$ possible combinations of signs: they have only $1$ allowed combination, namely the one exhibited by $\Psi$. On the contrary, the others $\lambda \log n - j$ are free to take any of the $2^{\lambda \log n - j}$ possible combinations of signs. The number of favourable outcomes can be thus considered to be the following:
\begin{equation}\label{eqn:pi_j_favourable}
2^{\lambda \log n - j} {n - s_{i-1} \choose {\lambda \log n - j}}{s_{i-1} \choose j}
\end{equation}
While the total number of possible outcomes corresponds to the total number of compatible clauses available, whichever amount of overlapping literals they might have with $\Psi$. Such number is just the summation, for all possible values of $j$, of the expression in \ref{eqn:pi_j_favourable}. Here it is:
\begin{equation}\label{eqn:pi_j_total}
\sum_{t = 0}^{\lambda \log n} 2^{\lambda \log n - t}  {n - s_{i-1} \choose {\lambda \log n - t}}{s_{i-1} \choose t}
\end{equation}
One precision about \ref{eqn:pi_j_favourable} and \ref{eqn:pi_j_total}: to be $100\%$ accurate, we would have to perturbate them a little bit. From the expression in $\ref{eqn:pi_j_total}$, we would have to subtract the number of compatible clauses already used so far to build $\Psi$, which is $i-1$. Whereas to the expression in $\ref{eqn:pi_j_favourable}$, we would have to subtract the expected number of already used compatible clauses having $j$ overlappings, which can be at most $i-1$ of course. We deliberately decide to ignore both these tiny adjustments, for the latter one would be very annoying to compute, would unnecessarily complexify the expression, and would gain us very little in terms of the accuracy of our estimation of $s_i$. As we will see in section \ref{sec:empirical}, doing such simplification has no palpable impact on predicting $s_i$. To be convinced of this also analytically, it is sufficient to observe how, as $n \to \infty$, that tiny amounts of at most $i-1 < \delta n$ clauses are swallowed up whole by the involved binomial coefficients. We are now ready to present the final expression for $s_i$, which is the following:  
\begin{equation}
\label{eqn:s_i}
s_i=\begin{dcases}
			0 & \text{if $i = 0$}\\
            \lambda \log n & \text{if $i = 1$}\\
            \sum_{j=0}^{\lambda \log n} ( s_{i-1} + \lambda \log n - j ) \frac{ 2^{\lambda \log n - j} { n-s_{i-1} \choose \lambda \log n - j }{ s_{i-1}\choose j } }{ \sum_{t=0}^{\lambda \log n} 2^{\lambda \log n - t} {n-s_{i-1} \choose \lambda \log n - t}{ s_{i-1} \choose t }} & \text{otherwise}
		 \end{dcases}
\end{equation}
\noindent See how it is a recursive definition, since $\forall i \geq 2$ it is the case that $s_i$ it depends on $s_{i-1}$. This is not surprising. See also how the case $i=1$ is indeed a special case of the general case $i \geq 2$: by plugging the former into the expression of the latter, we get back $\lambda \log n$. 
\\\\Let us now concentrate our attention on the computation of $p_i$. So we are at the point where $\Psi$ contains $s_i$ variables and $i$ clauses, and we would like to know which is the probability that a clause $c$ randomly picked from $\Phi$ is compatible with all the variables of $\Psi$, in order for $c$ to be enrolled in $\Psi$ as its $i+1$-th clause, without infringing its monotonicity. First of all, we observe the obvious fact that $p_0 = 1$: when $\Psi$ contains no clauses, any picked clause is certainly compatible with $\Psi$. For the general case where $i \geq 1$, we will determine $p_i$ as the ratio between favourable outcomes and total outcomes: the former being the number of available compatible clauses, the latter being the number of available clauses. As for the number of compatible clauses, we have already determined it in equation \ref{eqn:pi_j_total} while computing $s_i$, here we simply have to repropose it by taking care of using $s_i$ instead of $s_{i-1}$:
\begin{equation}
\sum_{t = 0}^{\lambda \log n} 2^{\lambda \log n - t}  {n - s_i \choose \lambda \log n - t}{s_i \choose t}
\end{equation}
As for the number of clauses in general, compatibles or not, here it is:
\begin{equation}
2^{\lambda \log n}  {n \choose \lambda \log n}
\end{equation}
From the number of compatible clauses, we may decide to subtract $i$, being it the number of compatibles clauses used so far, and thus no longer available. Analogously, from the number of clauses in general, we may subtract $w_i$, being it the number of original clauses of $\Phi$ wasted so far to arrive at this point where $\Psi$ has $i$ clauses. We let such tiny adjustments remain there, because both their expressions are straightforward now, but keeping in mind that what we said previously stands valid also in their case: as $n \to \infty$, they are steamrolled by the binomials. The final expression for $p_i$ is henceforth the following:
\begin{equation}
\label{eqn:p_i}
p_i=\begin{dcases}
			1 & \text{if $i = 0$}\\
            \frac{-i + \sum_{t=0}^{\lambda \log n} 2^{\lambda \log n - t} {n-s_i \choose \lambda \log n - t}{ s_i \choose t }}{2^{\lambda \log n} {n \choose \lambda \log n} - w_i } & \text{otherwise}
		 \end{dcases}
\end{equation}
\\We are finally ready to devote our attention to $w_i$. Here we are asking: knowing that $\Psi$ has $i$ clauses, how many original clauses of $\Phi$ did we have to dilapidate for us to arrive at such point? We start by observing the truism that $w_0$ = 0, i.e. to have an empty $\Psi$ we have to waste no clauses from $\Phi$. What about the general case $i \geq 1$? Well, there we have to waste $w_{i-1}$ clauses, plus the expected number of clauses of $\Phi$ needed to be overthrown before being able to find a compatible one to be added to $\Psi$ as its $i$-th clause. Such expected number is clearly $\frac{1}{p_{i-1}}$, which leads to the following:  
\begin{equation}
\label{eqn:w_i}
w_i=\begin{dcases}
			0 & \text{if $i = 0$}\\
            w_{i-1} + \frac{1}{p_{i-1}} = \sum_{r=0}^{i-1} \frac{1}{p_r} & \text{otherwise}
		 \end{dcases}
\end{equation}
Now that we have defined and computed $s_i$, $p_i$ and $w_i$, we are ready to end this section by formalizing the definition of (the expectation\footnote{Average number of clauses of any maximal monotone sub-formula $\Psi$ of $\Phi$.} of) $i_{STOP}$ which we informally introduced in the beginning:
\theoremstyle{definition}
\begin{definition}[$i_{STOP}$]
The smallest $i$ which satisfies $w_i \geq \delta n$.
\end{definition}
\noindent See how the $\delta$ parameter is, not surprising, completely absent from the expressions \ref{eqn:s_i}, \ref{eqn:p_i} and \ref{eqn:w_i}. It is precisely here, and only here, within the definition of $i_{STOP}$, that it comes into play. We close the section by displaying all our recurrence in a single frame:
\begin{equation*}
s_i=\begin{dcases}
			0 & \text{if $i = 0$}\\
            \lambda \log n & \text{if $i = 1$}\\
            \sum_{j=0}^{\lambda \log n} ( s_{i-1} + \lambda \log n - j ) \frac{ 2^{\lambda \log n - j} { n-s_{i-1} \choose \lambda \log n - j }{ s_{i-1}\choose j } }{ \sum_{t=0}^{\lambda \log n} 2^{\lambda \log n - t} {n-s_{i-1} \choose \lambda \log n - t}{ s_{i-1} \choose t }} & \text{otherwise}
		 \end{dcases}
\end{equation*}
\begin{equation*}
p_i=\begin{dcases}
			1 & \text{if $i = 0$}\\
            \frac{-i + \sum_{t=0}^{\lambda \log n} 2^{\lambda \log n - t} {n-s_i \choose \lambda \log n - t}{ s_i \choose t }}{2^{\lambda \log n} {n \choose \lambda \log n} - w_i } & \text{otherwise}
		 \end{dcases}
\end{equation*}
\begin{equation*}
w_i=\begin{dcases}
			0 & \text{if $i = 0$}\\
            w_{i-1} + \frac{1}{p_{i-1}} = \sum_{r=0}^{i-1} \frac{1}{p_r} & \text{otherwise}
		 \end{dcases}
\end{equation*}
\begin{equation*}
i_{STOP} = min \{\ i \ |\ w_i \geq \delta n\ \}
\end{equation*}
 
\section{Proof of sub-exponential running time}\label{sec:proof}
\noindent We will now prove that $i_{STOP} \in o(n)$. In order to do that, we are not going to unfurl the accurate recurrence above so to flat it down into a single, non-recurring expression. Such recurrence will however play an important role in Section \ref{sec:empirical}. Instead, here we will make some simplifications that will render us able to obtain a rough upper bound on the growth rate of $i_{STOP}$. None of such simplifications will jeopardise the validity of our conclusion, though: for they will only worsen the quality of our upper bound, which nevertheless will remain $o(n)$ in the end, so sufficient for our purpose.
\newpage
\begin{theorem}
$i_{STOP} \in o(n)$
\end{theorem}
\begin{proof}
Let us start by displaying the following simple drawing:
\begin{figure}[H]
    \centering
    \includegraphics[scale=0.15]{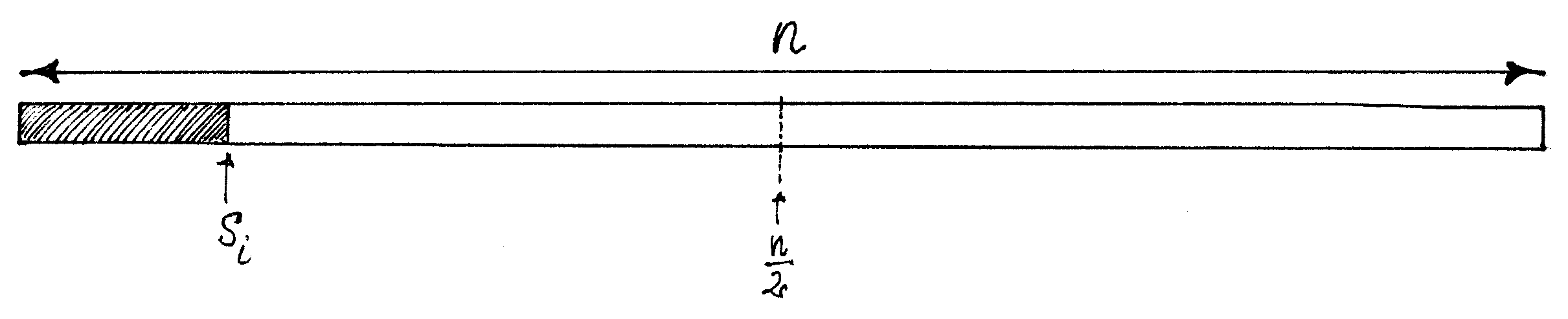} \\
    \caption{$s_i$ growing towards $\frac{n}{2}$ as we keep adding clauses to $\Psi$}
    \label{fig:s_i_growing}
\end{figure}
\noindent We have to imagine to be at any early moment of the process of building $\Psi$: there are $i$ clauses in it, and it has $s_i < \frac{n}{2}$ variables. We observe the obvious fact that, on average, each time we add a compatible clause of $\lambda \log n$ literals to $\Psi$, at least $\frac{\lambda \log n}{2}$ of them will be \textit{new} variables for $\Psi$: basically all those who fall into the "right" half of the $n$ variables, which is still completely empty. Actually such new variables will be more than that, because there will be a part of them, less and less as $s_i$ grows, picked from the "left" half of $n$. This simple fact is due to the random nature of the picked clause: on average half of its literals will fall on the left side, and half on the right side. That is the first simplification we make: see how doing so implies we consider $s_i$ as growing slower than it actually does in reality, thus we are only raising the bar of our upper bound. So, as we are certain to add at least $\frac{\lambda \log n}{2}$ new variables each time we add a new compatible clause to $\Psi$, we can be certain that after $\frac{n}{2}\frac{2}{\lambda \log n} = \frac{n}{\lambda \log n}$ clauses added, we have to have at least $s_i = \frac{n}{2}$ crystalised variables:
\begin{figure}[H]
    \centering
    \includegraphics[scale=0.15]{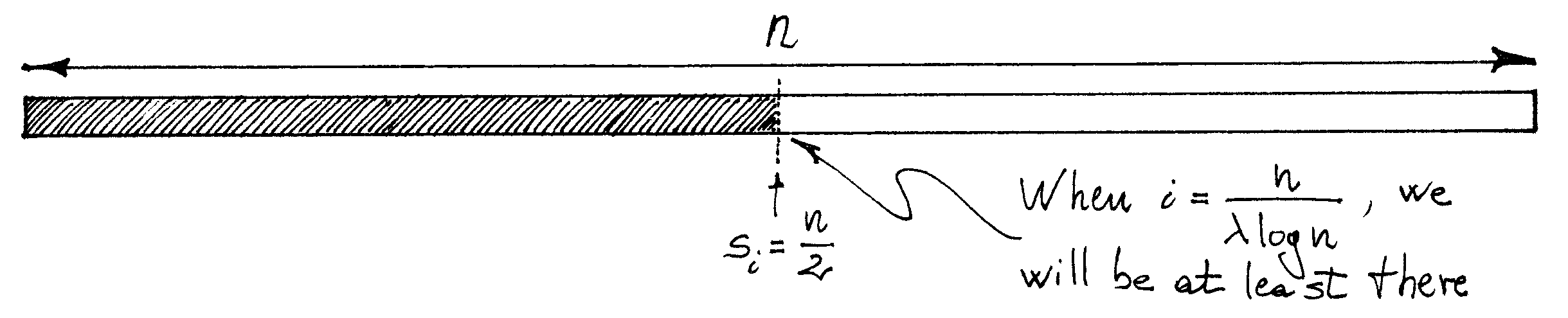} \\
    \caption{When $\Psi$ has $i = \frac{n}{\lambda \log n}$ clauses, it also has at least $s_i = \frac{n}{2}$ variables}
    \label{fig:s_i_n_2}
\end{figure}
\noindent Such basic statement, which is quite obvious on its own, has nevertheless been subject to empirical fact checking: the evidence in support of it is so overwhelming that it leaves no room, if ever has there been any, for believing the opposite. Actually, as we easily predicted, the empirically observed $s_i$ was much higher than $\frac{n}{2}$. Now that we are at this point, let us ask the following question: which is the probability that a clause $c$ randomly picked from $\Phi$ will be compatible with $\Psi$ now that it has $s_i = \frac{n}{2}$ variables? Well also in this case, as usual, on average $\frac{\lambda \log n}{2}$ literals of $c$ will fall on the "right" half of $n$, which is still completely empty, whereas the others $\frac{\lambda \log n}{2}$ will fall on the "left" half, which is now completely filled. The probability that the literals fallen on the left side will be \textit{all} compatible with the $s_i$ crystalised variables is:
\begin{equation*}
\frac{1}{2^{\frac{{\lambda log n}}{2}}} = \frac{1}{\sqrt{n^\lambda}}
\end{equation*}
This means that, once $s_i = \frac{n}{2}$ is reached, from there onwards each clause we will add to $\Psi$ will require us to waste at least $\sqrt{n^\lambda}$ clauses from $\Phi$. Actually, and here is the second simplification we make, such wasting will grow higher and higher as more and more clauses will be added to $\Psi$, but we will pretend that such wasting remains fixed at $\sqrt{n^\lambda}$: again, this is not going to be an issue for our line of reasoning, as it will only raise the bar of the upper bound we are computing, worsening its quality (for in our simplified analysis the clauses of $\Phi$ will be all dilapidated after it actually happens in reality).
\newpage
\noindent At this point, we only have to determine how many further clauses can we add to $\Psi$ before we run out of original clauses to be picked from $\Phi$. We start from the $i = \frac{n}{\lambda \log n}$ clauses which allowed us to reach $s_i = \frac{n}{2}$: let us pretend that each of these clauses costed us the wasting of only $1$ clause from $\Phi$, that is to say the minimum waste possible. This is obviously a gross underestimation of reality, and it is our third simplification: it will make us pretend to have much more remaining clauses of $\Phi$ at our disposal than there actually are (and this will raise even more the bar of our upper bound). So we have 
\begin{equation*}
\delta n - \frac{n}{\lambda \log n}
\end{equation*}
remaining clauses in $\Phi$, from where to pick our candidate clauses to be potentially added into $\Psi$. Considering that we are forced to waste at least $\sqrt{n^\lambda}$ of them \textit{for each} new clause that we want to add to $\Psi$, it follows that the number of further clauses that we might hope to add to $\Psi$ is at most
\begin{equation*}
\frac{\delta n - \frac{n}{\lambda \log n}}{\sqrt{n^\lambda}}
\end{equation*}
before we exhaust all the $\delta n$ clauses at our disposal in $\Phi$. By adding to such number the  $\frac{n}{\lambda \log n}$ clauses initially added to $\Psi$ that allowed us to reach $s_i = \frac{n}{2}$, we get
\begin{equation}\label{eq:upper}
i_{STOP} \leq \frac{n}{\lambda \log n}(1-\frac{1}{\sqrt{n^\lambda}}) + \frac{\delta n}{\sqrt{n^\lambda}}
\end{equation} 
from which it immediately follows that $i_{STOP} \in o(n)$.
\end{proof}
\noindent Let us now look back at the pseudocode of our algorithm: we have to merely enumerate all the existing monotone sub-formulae of $\Phi$, which boils down to enumerate all the possible subsets of clauses of $\Phi$ having cardinality at most $i_{STOP}$. How many such subsets are there? Considering that $i_{STOP} \in O( \frac{n}{\lambda \log n} )$, their number is upper bounded by how many ways are there to pick $\frac{n}{\lambda \log n}$ clauses from the initial ensemble of $\delta n$ clauses. We are ready to conclude by proving the following:
\begin{theorem}
The procedure $\textsc{Count}(\Phi)$ runs in sub-exponential time in the size of $\Phi$.
\end{theorem}
\begin{proof}
We only have to show the truth of the following intuitively basic fact\footnote{This short proof was suggested by Qiaochu Yuan, \url{https://math.stackexchange.com/q/3829594/10194}}:
$${\delta n \choose \frac{n}{\lambda \log n}} \in 2^{o(n)}$$
Stirling's formula tells us that, asymptotically, the following holds:
$$\log { a \choose b } \approx b \log \frac{a}{b} + (a-b) \log \frac{a}{a - b}$$
By plugging into it $a = \delta n$ and $b = \frac{n}{\lambda \log n}$, we obtain
$$\log { \delta n \choose \frac{n}{\lambda \log n} } \approx \frac{n}{\lambda \log n} \log ( \delta \lambda \log n ) - ( \delta n - \frac{n}{\lambda \log n} ) \log ( 1 - \frac{1}{\delta \lambda \log n} )$$
As $\log( 1 - x ) \approx -x$ for small $x$, we can use it on the rightmost term to obtain:
$$\log { \delta n \choose \frac{n}{\lambda \log n} } \approx \frac{n}{\lambda \log n} \log ( \delta \lambda \log n ) + \frac{n}{\lambda \log n} - \frac{n}{\delta \lambda^2 \log^2 n}$$
where the leftmost term can be easily seen to be the leading term. The asymptotic running time of the $\textsc{Count}(\Phi)$ procedure is therefore the following:
%$$2^\frac{n \log ( \delta \lambda \log n )}{ \lambda \log n} \in 2^{o(n)}$$
\[\Scale[1.75]{2^\frac{n \log ( \delta \lambda \log n )}{ \lambda \log n} \in 2^{o(n)}}\]
\end{proof}

% These commands need to appear at the point where you want
% the first page to end.
\restoregeometry
\newgeometry{bottom=0.5in}

\section{Empirical evidence in further strengthening of the theoretical evidence}\label{sec:empirical}
\noindent The overall objective of our empirical investigation is visually described by the soldering of the following $2$ complementary faces of the same medal, which are going to be addressed in the next $2$ subsections:
\begin{itemize}
    \item  To increase confidence that \ref{eq:upper} truly is an upper bound on $i_{STOP}$.
    \item  To increase confidence that $i_{STOP}$ truly is an accurate representation of reality.
 \end{itemize}
Aim is to combine such $2$ facts in order to reach the desired conclusion: \ref{eq:upper} is a valid upper bound of reality.

\subsection{Empirical correctness of sub-linear upper bound} In order to empirically ascertain that $i_{STOP}$ growth rate is upper bounded by \ref{eq:upper}, we have written the software implementation of the recurrence described in section \ref{sec:intuition}. We wrote a function \texttt{predict}( $n$, $\delta$, $\lambda$ ) which returns $i_{STOP}$ by unfurling our recurrence, from $i = 0$ up to the first $i$ such that $w_i \geq \delta n$, which is the returned value. We invoked such function up to more than $4$ billions variables, $n = 2^{32} +1$ to be precise. We tested several combinations of the $\delta$ and $\lambda$ parameters, roughly falling into $3$ distinct categories: fixed $\lambda = 1$ and growing $\delta$, fixed $\delta = 1$ and growing $\lambda$, plus finally the $4$ combinations induced by $\delta \in \{2, 2048\}$ and $\lambda \in \{2,4\}$. The results are illustrated in the following $3$ figures. Needless to say, these collected data leave no hope, if ever has there been any, for believing that \ref{eq:upper} is not a valid upper bound on $i_{STOP}$.
% 1. When figure position is crucial, use the [H] tag 
%    to enforce absolute positioning
% 2. Image files should be referenced without file extension
\begin{figure}[H]
    \centering
    \makebox[\textwidth][c]{\includegraphics[width=1.35\textwidth]{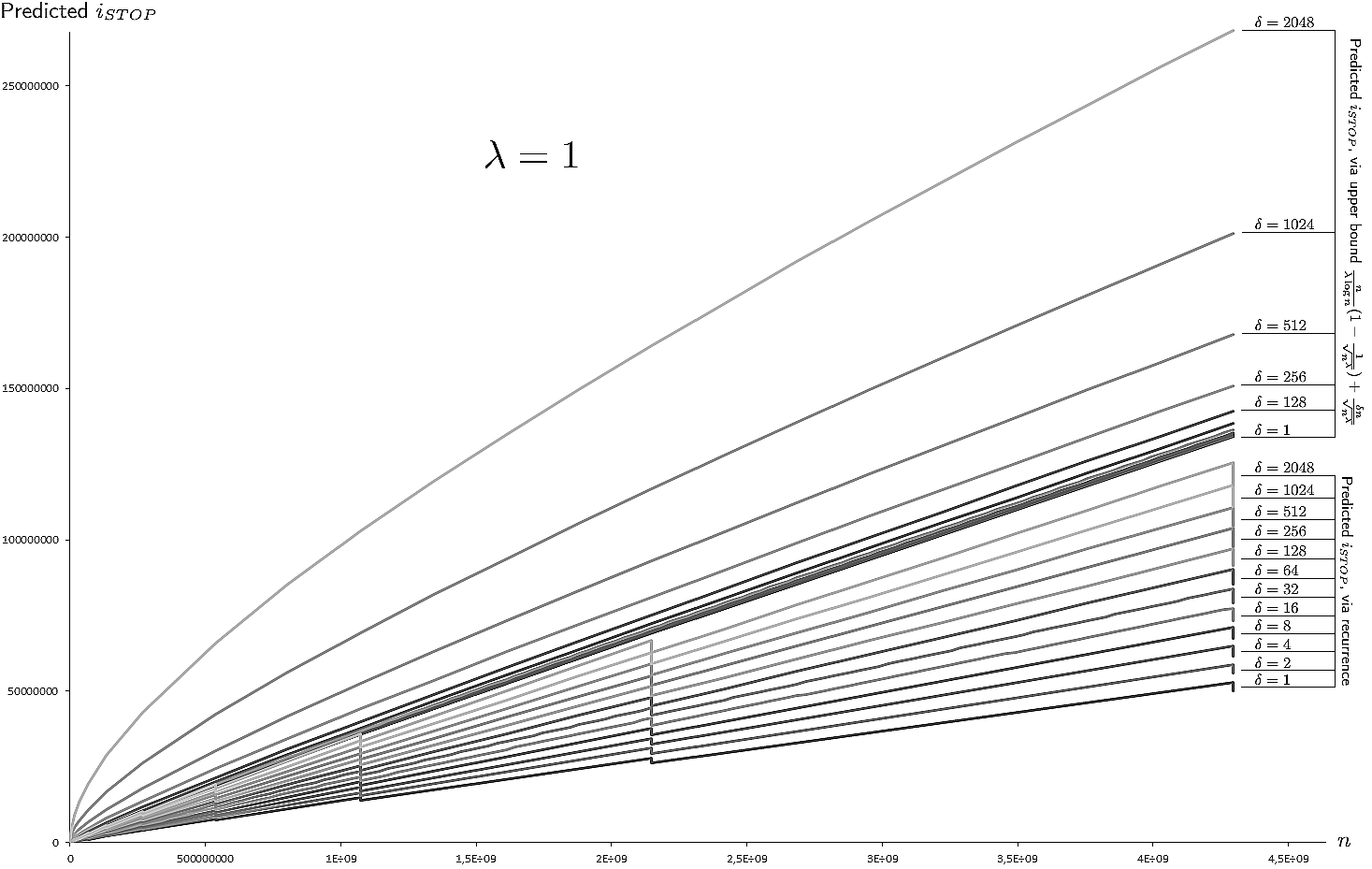}}
    \caption{Comparison between $i_{STOP}$ predictions, for fixed $\lambda = 1$ and growing $\delta$}
    \label{fig:recurrence_upper_bound}
\end{figure}
\newpage 
\begin{figure}[H]
    \centering
    \makebox[\textwidth][c]{\includegraphics[width=1.29\textwidth]{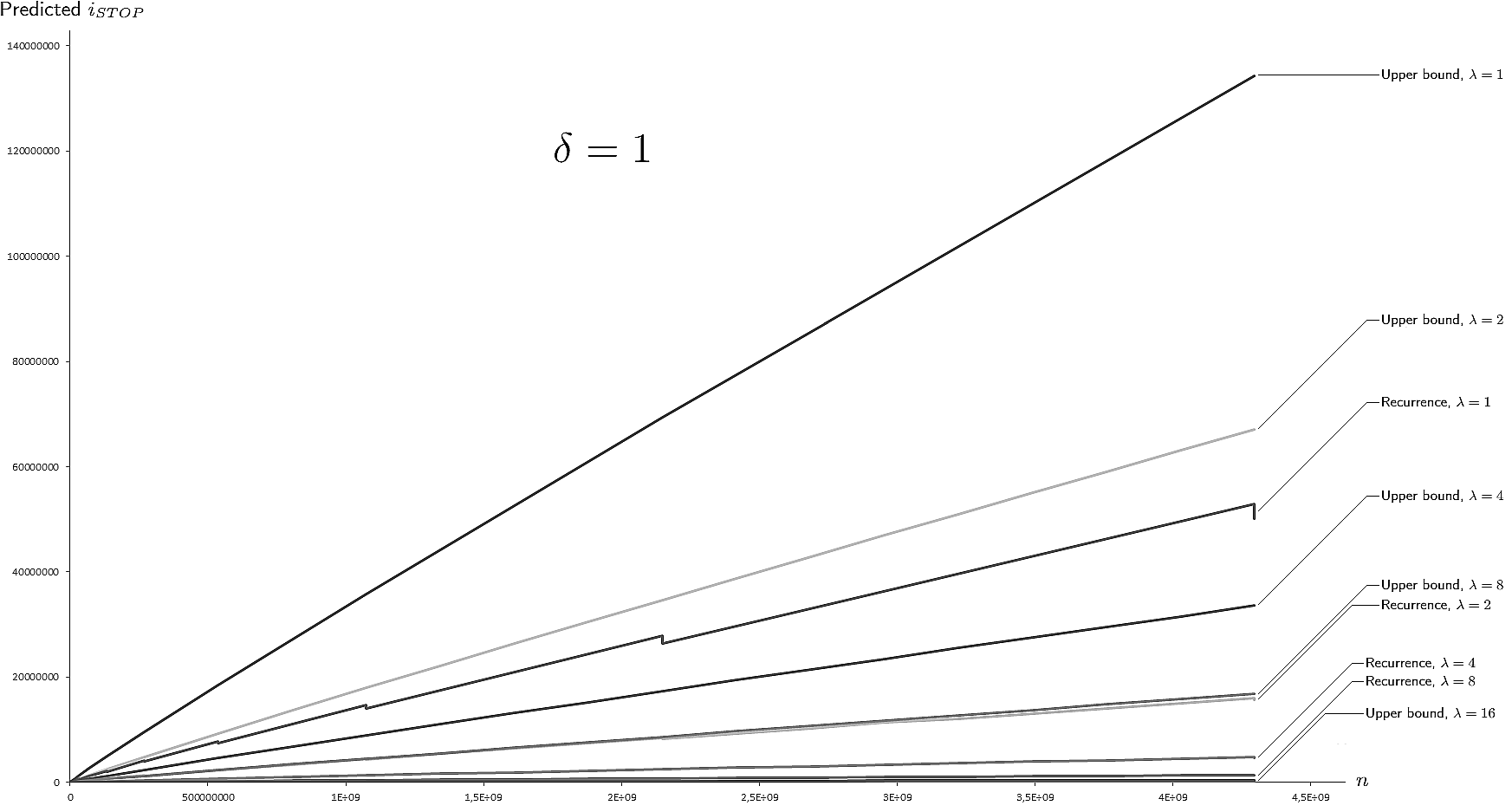}}
    \caption{Comparison between $i_{STOP}$ predictions, for fixed $\delta = 1$ and growing $\lambda$}
    \label{fig:recurrence_upper_bound_2}
\end{figure}

\begin{figure}[H]
    \centering
    \makebox[\textwidth][c]{\includegraphics[width=1.285\textwidth]{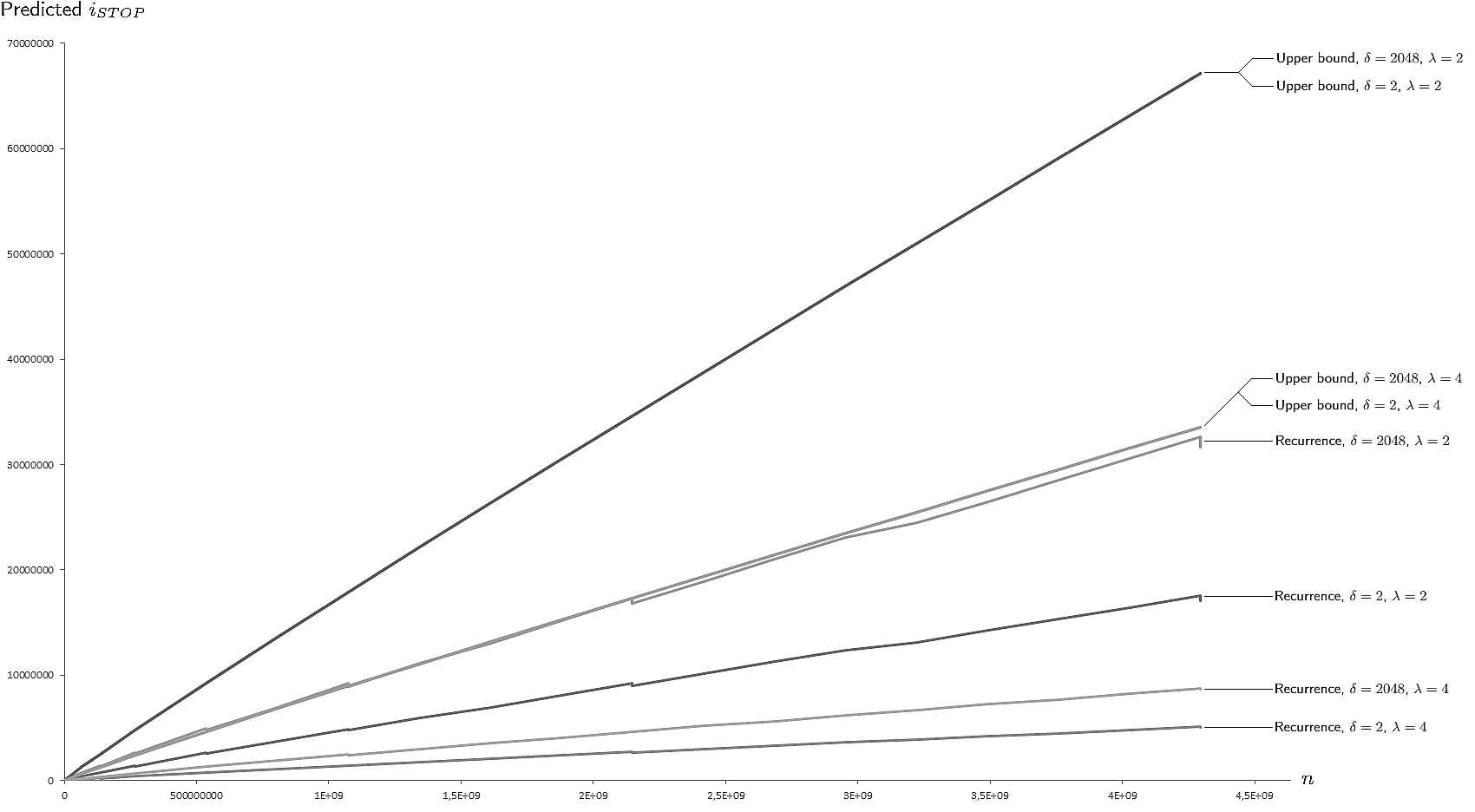}}
    \caption{Comparison between $i_{STOP}$ predictions, for $\delta \in \{2,2048\}$ and $\lambda \in \{2,4\}$}
    \label{fig:recurrence_upper_bound_3}
\end{figure}

\subsection{Empirical correctness of equations \ref{eqn:s_i}, \ref{eqn:p_i} and \ref{eqn:w_i}} We begin directly by showing the figures:

\begin{figure}[H]
    \centering
    \makebox[\textwidth][c]{\includegraphics[width=1.35\textwidth]{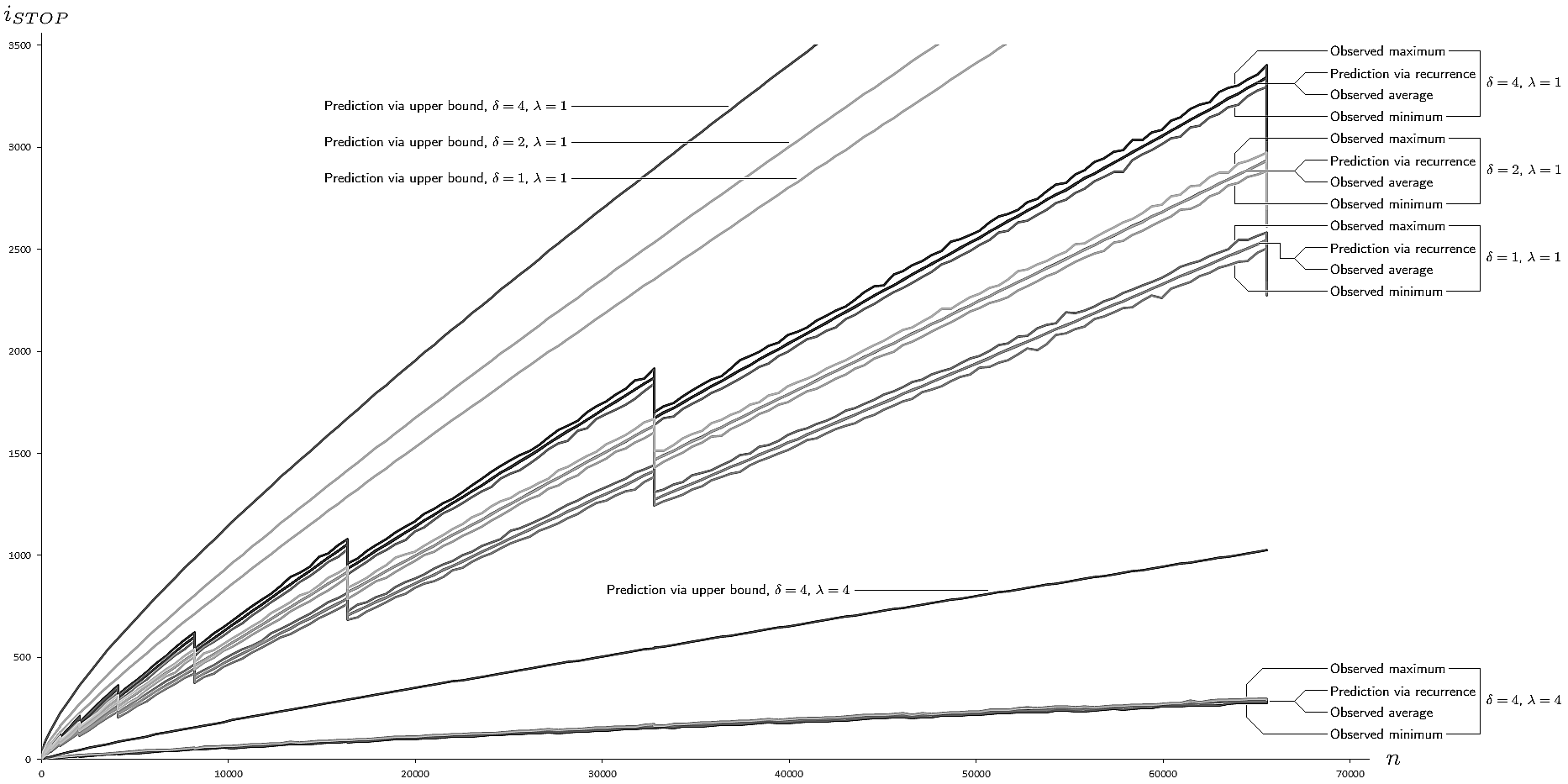}}
    \caption{Fact checking $i_{STOP}$ prediction accuracy against reality, for $n \leq 2^{16} + 1$}
    \label{fig:reality}
\end{figure}

\begin{figure}[H]
    \centering
    \makebox[\textwidth][c]{\includegraphics[width=1.36\textwidth]{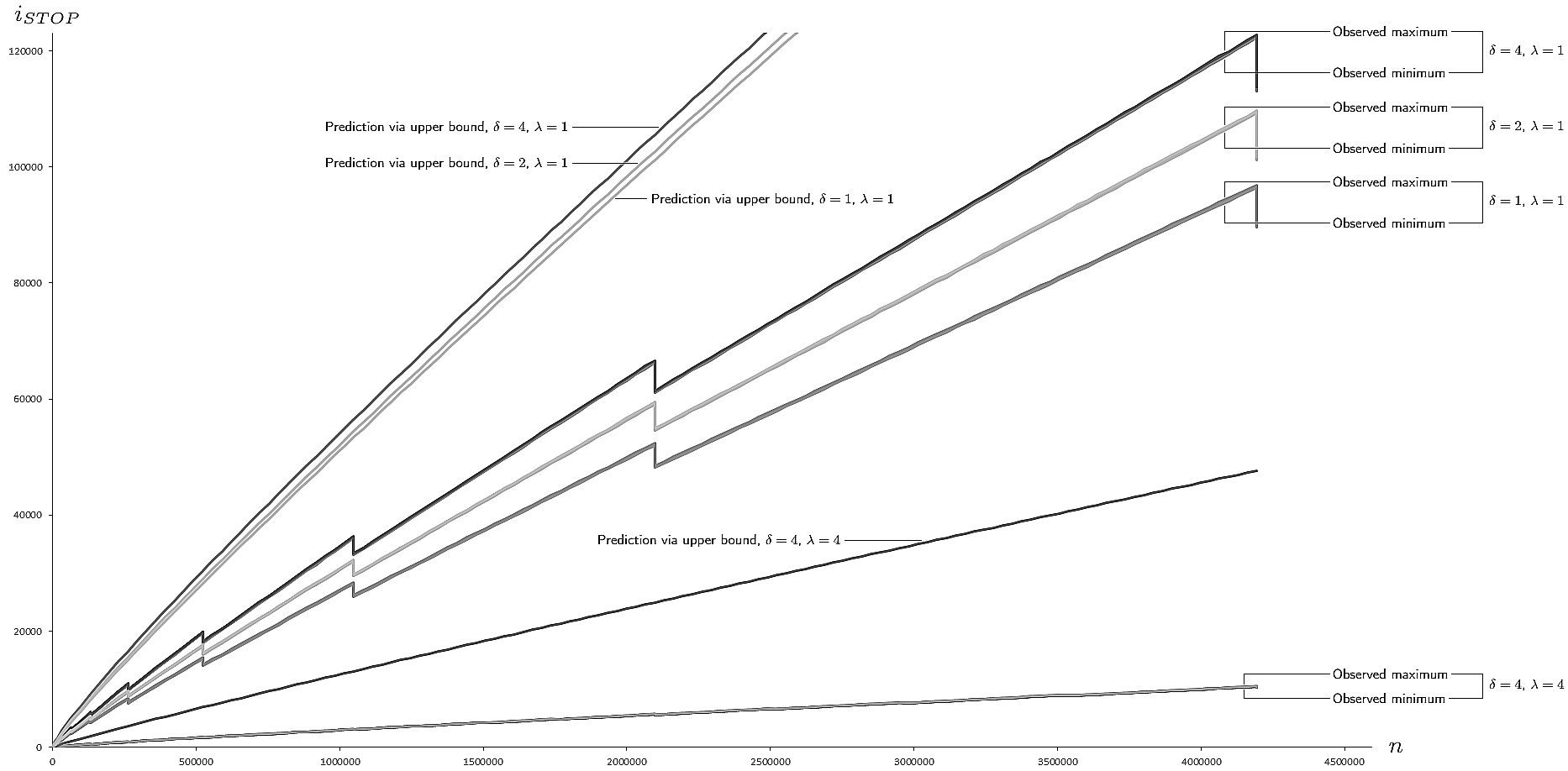}}
    \caption{Fact checking $i_{STOP}$ prediction accuracy against reality, for $2^{16}-1 \leq n \leq 2^{22} + 1$}
    \label{fig:reality_4m}
\end{figure}
\newpage
\noindent In order to empirically ascertain that the formalization of $i_{STOP}$ given in section \ref{sec:intuition} is an accurate representation of true, factual reality, we have written the software implementation of the process we have described therein. We wrote a function \texttt{psi}( $n$, $\delta$, $\lambda$ ) which returns the number $i$ of clauses of a randomly picked maximal monotone sub-formula $\Psi$ of a randomly generated formula $\Phi$ having $n$ variables, $\delta n$ clauses, and $\Lambda_\downarrow = \Lambda_\uparrow = \lambda \log n$. We invoked such function up to more than $4$ millions variables, $n = 2^{22}+1$. For each combination of $n$, $\delta$ and $\lambda$ which is shown in the figures, the \texttt{psi} function was invoked $100$ times. The empirically observed degree of accuracy of the $i_{STOP}$ value predicted by the recurrence goes beyond our most optimistic expectations, as the predicted average and the observed average can be considered coincident from any practical standpoint: see how, in figure \ref{fig:reality}, their plots are indistinguishable one another. Observe also how in figure \ref{fig:reality_4m}, where the number of variables is remarkably higher than in figure \ref{fig:reality}, the observed minimum and maximum values become kind of coincident as well, and practically no longer distinguishable, sandwiching the predicted and observed averages between them. Conclusion is, needless to say, that these observed data leave us with not even a minuscule bit of possibility, if ever has there been any, for believing that our formalization of $i_{STOP}$ is not a trustworthy description of reality. 

\subsection{Availability of data and software implementation}
All the empirically gathered data are available for download, as a spreadsheet, at the following web address: \url{http://gcamerani.altervista.org/subexp}. Within the same web page, it will be possible to download the Java software that we implemented\footnote{Our code infringes almost every mandate of object-oriented programming and good software design. We \textit{count} on your comprehension: we had to write it quickly, and it had to be fast.} to unroll our experimental investigation, i.e. the \texttt{predict} and \texttt{psi} functions. We also implemented the algorithm itself, in order to empirically check the validity of equation \ref{eq:property} and to compare it to other model counters\footnote{To get a glimpse of what are we talking about: when presented a certain $\Phi$ with $n = 68$, $\delta = 1$ and $\lambda = 1.2$, our algorithm answered $|\mathcal{S}| = 226243273496345990400$ in $0.51$ seconds and $|\Phi|$ memory footprint, whereas \texttt{c2d} delivered the same answer after $2036$ seconds and using more than $800$ Mb of memory. With $n = 70$, ours delivered $|\mathcal{S}| = 897869659263845943936$ in $0.842$ seconds and $|\Phi|$ memory (which is always the case of course), while \texttt{c2d} finished in $4$ hours, $2$ minutes and $24$ seconds, with around $1.5$ Gb memory usage.}.

\section{Final considerations and possible avenues of future research}
\noindent We have presented a new paradigm which allows us to count satisfying assignments without even searching for them, let alone finding them. We have shown that, when applied to random sparse formulae, such paradigm leads us to the following unavoidable conclusion: the counting algorithm runs faster and faster as the minimum clause length gets higher and higher. The same holds for the counting certificate of course, which shrinks and shrinks: for the counting certificate itself is just the whole sub-exponentially sized space of monotone sub-formulae, which cardinality decreases as the clause length increases. We have seen how, when the clause length is at least logarithmic in $n$, this leads to a sub-exponential time algorithm. %This fact immediately leads to recalling a nice observation made at page $2$ of \cite{fasterSAT}, where they notice how the existence of an algorithm for $k$-SAT in $n$ variables and $f(n)$ clauses which runs in time $2^{sn}$ for some constant $s<1$ and where $k \geq \frac{1}{s}\log f(n) + \Omega(1)$ would imply the existence of an improved algorithm for CNF-SAT in $n$ variables and $f(n)$ clauses. They observe how it would thus be sufficient to restrict attention to faster algorithms able to deal with logarithmic length clauses. The algorithm described in this paper is precisely such an algorithm, where $f(n) = \delta n$, except that instead of running in time $2^{sn}$ for some constant $s<1$, it runs in sub-exponential time. Such observation in \cite{fasterSAT} consists in a specific instantiation of Lemma $5$ in \cite{clauseWidth}, which states that if $k$-SAT can be solved in time $O(2^{sn})$ by some algorithm $S$ then CNF-SAT in $n$ variables and $m \geq \frac{n}{k}$ clauses can be solved in time $\text{poly}(n,m)2^{sn + \frac{4m}{2^{sk}}}$, where the polynomial does not depend on $k$, $s$, nor $S$. By applying such Lemma with $k = \lambda \log n$, $m = \delta n$ and with $s = \frac{\log( \delta\lambda\log n )}{\lambda \log n}$ obtained in this paper, an improved algorithm running in time $\text{poly}(n)2^{\frac{\log(\delta\lambda\log n) + 4}{\lambda \log n}n}$ follows for random CNF-SAT in $n$ variables and $\delta n$ clauses, without any restriction on the minimum clause length. This consequence can be extended to random dense CNF-SAT as well, thanks to the existence of the Sparsification Lemma introduced in \cite{sparsification}: any such dense CNF-SAT can be disintegrated, in sub-exponential time, in sub-exponentially many sparse CNF-SAT such that the original dense formula is satisfiable if and only if at least one the sparse formulae is.
\\\\
Let us thus perform a further, quick final reflection along this line. We observe how the reasoning described in section \ref{sec:proof} would seamlessly go through whatever your favourite minimum clause length $\Lambda_\downarrow$ is. See how this implies that, when $\Lambda_\downarrow = \frac{n}{\mu}$ where $\mu \geq 1$ is a constant, the algorithm runs in asymptotic polynomial time $\Lambda_\uparrow (\delta n)^{\mu}$, as long as $\delta$ is polynomial in $n$. Whereas on the opposite side, when $\Lambda_\downarrow = k$ where $k > 1$ is a constant, the algorithm runs in exponential time $O(2^{\varepsilon n})$, where $\varepsilon = \mu \log \frac{\delta}{\mu}$ and $\mu = \frac{1}{k} + \frac{\delta - \frac{1}{k}}{\sqrt{2^k}}$, which essentially means $\lim_{k \to \infty} \varepsilon = 0$ as long as $\delta \in 2^{o(k)}$.
\\\\
We just stop for a moment, make a step back and look at this again: the sub-exponential running time proof presented in section \ref{sec:proof} would effortlessly work as well, without any modification, for any non-constant $\Lambda_\downarrow$. This means that the $\Lambda_\downarrow = \lambda \log n$ assumption which we stated since line $3$ of the introduction, and around which this entire paper has been forged, is just an arbitrary choice among the infinitely many non-constant functions $\Lambda_\downarrow$. Any monotonically increasing function $\Lambda_\downarrow \in \omega(1)$, no matter how ridicolously slowly growing it might be, would do the job: in the end, the running time we are going to obtain will always be $O(2^{\varepsilon n})$, where $\varepsilon = \frac{\log \Lambda_\downarrow}{\Lambda_\downarrow}$. All these last observations lead us to ponder about ETH, which we are going to rapidly address in the following subsection.

\subsection{Impact on ETH} The Exponential Time Hypothesis, originally formulated by Impagliazzo and Paturi in \cite{ETH}, postulates that $k$-SAT cannot be solved in worst case sub-exponential time for every $k \geq 3$. When dealing with ETH, these $2$ assumptions are normally made: $k$ is a constant, and the number $m$ of clauses is linear in the number $n$ of variables. Such second assumption can be made due the existence of the Sparsification Lemma introduced in \cite{sparsification}, thanks to which any $k$-SAT instance having $m \in \omega(n)$ can be disintegrated, in sub-exponential time, in sub-exponentially many sparse $k$-SAT instances such that the original dense formula is satisfiable if and only if at least one of the sparse formulae is. So ETH belongs to the constant $k$ realm, whereas our algorithm counts satisfying assignments in sub-exponential time only in the non-constant $k$ realm; moreover ETH does not explicitly require that the instance is random, whereas for our algorithm to work this is a crucial assumption; finally our algorithm counts the exact number of satisfying assignments, whereas ETH is about existence of at least one satisfying assignment (we observe the presence of \#ETH, the counting equivalent of ETH, introduced in \cite{CountETH}). There are quite some differences, that one may argue whether ETH and our algorithm have anything to say each other at all. As the reader might have guessed, it turns out that there is something indeed. Importantly, if ETH is true then, as proven in \cite{ETH}, the coefficient of the exponent in the running time, $s_k$ to adopt the notation therein, strictly increases infinitely often. That is to say, assuming ETH the problem gets harder and harder as $k$ grows. But this paper has just shown that actually this is \textit{not always} the case. Let us visualize it more clearly:
\begin{figure}[H]
    \centering
    \makebox[\textwidth][c]{\includegraphics[width=0.8\textwidth]{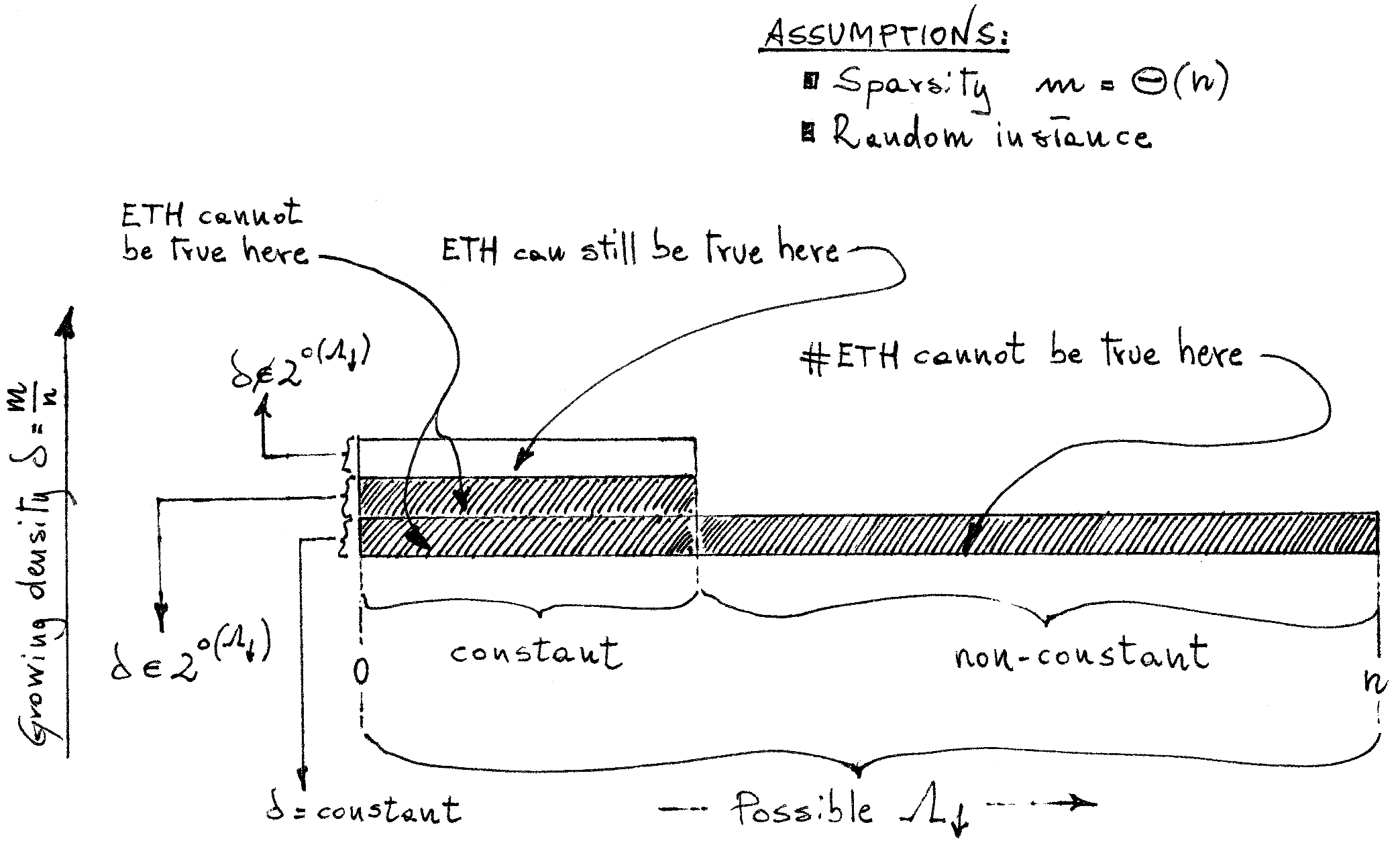}}
    \caption{What the existence of our algorithm has to say in relation to ETH and \#ETH}
    \label{fig:ETH}
\end{figure}
\noindent We begin by observing how, in the non-constant $k$ realm, the \#ETH cannot be true, by the very fact that our algorithm exists. See how such realm is uninteresting in terms of the ETH point of view, because for the instances over there, a satisfying assignment can be easily found. While in the constant $k$ side, $2$ paragraphs above we observed how, as long as $\delta \in 2^{o(k)}$, the problem gets easier and easier as $k$ grows, and our algorithm runs in time $O(2^{\varepsilon n})$ with $\varepsilon$ approaching $0$ as $k$ approaches $\infty$. This means that, in the sub-region where the density $\delta = \frac{m}{n}$ is sub-exponential in $k$, the ETH is false (as well as the \#ETH): for its consequences would disobey reality, as here the sequence of $s_k$ is decreasing rather than increasing. While in the region where the density is at least exponential in $k$, the ETH can be possibly true: see how this is consistent with current knowledge, according to which random instances having density exponential in $k$ are the hardest ones existing. We conclude by observing how, in order to believe ETH now, we must believe that as $k$ grows, $k$-SAT gets easier and easier as long as the density is sub-exponential in $k$, and the $s_k$ sequence keeps on decreasing monotonically, then, as soon as we tune the density higher and step into the sub-region exponential in $k$, keeping $k$ growing and running the \textit{best} existing algorithm therein, such trend is surprisingly inverted all of a sudden for some mysterious reason: a $180$ degrees turn occurs and the sequence commences to inexorably increase toward $1$.  %It's just that, if ETH is true over there, the following pattern would have to hold as $k$ grows: the coefficient of the exponent monotonically decreases, then it inverts the trend and starts to increase infinitely often, then it inverts the trend again.

%Secondly, one observation in relation to sparsity. Throughout the whole paper we have assumed that the number of clauses of $\Phi$ is linear in the number $n$ of its variables via the constant $\delta$:
%-------------------
%it is easy to see, however, that our conclusion for $\Lambda_\downarrow \geq \lambda \log n$ can seamlessly go through as well for random formulae having any polynomial number $n^\alpha$ of clauses, provided that $\lambda \geq 2 \alpha$.
%$O(2^{\frac{\log (\delta k)}{k} n})$
%-------------------
%it is easy to see, however, that for random formulae having any polynomial number $n^\alpha$ of clauses, and assuming $\Lambda_\downarrow \geq \lambda \log n$ as usual, the resulting exponential running time $O(2^{\varepsilon n})$, with $\varepsilon = \frac{\alpha - 1}{\lambda}$, can be made arbitrarily small by choosing a sufficiently large $\lambda$; whereas if we are willing to set $\Lambda_\downarrow \geq \lambda \log^2 n$, the running time would be back to sub-exponential again, whatever $\alpha$ and $\lambda$ are, that is to say $O( 2^{\frac{\alpha - 1}{\lambda \log n}n} )$.
%-------------------
%by applying the very same reasoning unrolled in section \ref{sec:proof} it is easy to see, however, that assuming $k = \Lambda_\downarrow = \Lambda_\uparrow \geq \lambda \log n$, for random formulae having non-constant $\delta \in O(\frac{}{})$ the running time of the algorithm would stand sub-exponential.
\subsection{Possible directions of future research}
\noindent We conclude the paper by suggesting $3$ promising avenues to be investigated in the very near future. Firstly, we would suggest to carefully analyse the asymptotic behaviour of the recurrence we presented in section \ref{sec:intuition}, in order to determine a more precise upper bound than ours. %: we believe that any professional mathematician willing to concentrate on such task would be able to crack the recurrence down into a flat, non recurring, \textit{more accurate} expression in few working hours, if not minutes.
Why we believe this is such an important and urgent task? Because if $i_{STOP}$ can be proven to belong to $o(\frac{n}{\log n})$, then another more general algorithm would immediately follow. What do we mean by "more general"? We mean a sub-exponential time algorithm which doesn't require any restriction on the minimum clause length, nor the random restriction.
\\\\
Secondly, in the same spirit but with a more pragmatic approach, we would advice to try to shape an algorithm for random $3$-SAT based on the algorithm described here: the basic idea would be to carefully resolve several variables, by applying resolution steps to the input formula, until the minimum clause length becomes long enough (whatever that means, the longer the better) for this algorithm to be invoked with enough profit. Clearly, this process would blow-up the number of clauses, but if an adequate equilibrium between proliferation of clauses and clause elongation can be found, then maybe that would result into a new algorithm with an interesting running time: in the most pessimistic scenario, the resulting running time might be no less than exponential, still possibly with a nice $\varepsilon$ compared to the state of the art; in the most optimistic scenario, if the clause length can be pushed to linear, i.e. $\frac{n}{\mu}$ for some constant $\mu > 1$, by keeping the number of clauses polynomially bounded, that would result in a polynomial time algorithm (such a taming down of the number of clauses should not be dismantled as an absurdity without prior empirical scrutiny, because as the number of clauses grows together with their length, the probability to obtain tautological clauses during resolution steps gets higher and higher, thus lots of them would be overthrown).
\\\\
Thirdly, the more promising avenue, we would be interested to empirically study the space of monotone sub-formulae of random formulae, in order to acquire deeper understanding of how the members of such space are distributed along the summation of equation \ref{eq:rearrangement}, and of how precisely they affect it. See how such summation can be understood as a finer and finer approximation of $|\mathcal{U}|$ as the index $\nu$ grows, until it reaches $\nu = n$ at which point the returned value becomes the exact value of $|\mathcal{U}|$. What if we cut the summation at some index $\nu < n$? Think about this: the value of $|\mathcal{U}|$ requires $n+1$ bits to be represented, where the $n+1$-th bit (the most significant) is the \textit{unsatisfiability} bit. Such bit is $1$ if and only if the formula is unsatisfiable. Intuitively, it sounds conceivable that, in random instances (where the absolute value of the quantity $O_\nu - E_\nu$ should remain small, whatever that means), the $n+1$-th bit gets frozen at some index $\nu < n$, and never changes anymore at higher indexes. That is to say, the value of $O_\nu - E_\nu$ for higher indexes would not have a sufficient magnitude to affect any longer the most significant bits already computed up to that point, and its bits would kind of cancel out with the corresponding bits of opposite sign belonging to others $O_\nu - E_\nu$ values at neighbouring indexes. If this is indeed the case, then we would be able to determine satisfiability by unrolling the summation only up to the index at which the $n+1$-th bit gets frozen, avoiding to compute the subsequent terms. If such frozening index is constant, that would mean polynomial time. And see also how the $n$-th bit of $|\mathcal{U}|$ is the negation of the \textit{majority} bit: it is equal to $0$ if and only if $|\mathcal{S}| > \frac{1}{2}2^n$. Just speculating with intuition...

\printbibliography%[heading=bibintoc,title={WHOLE}] %Prints the entire bibliography with the titel "Whole bibliography"

\end{document}